\documentclass[11pt]{article}

\usepackage{amssymb} 
\usepackage{amsmath}     
\usepackage{amsthm}
\usepackage{todonotes}
\usepackage{mathtools}

\usepackage{a4wide}
\usepackage{graphicx}
\usepackage{hyperref}
\newcommand{\Rset}{\mathbb{R}}
\newcommand{\Zset}{\mathbb{Z}}
\newcommand{\Qset}{\mathbb{Q}}

\newtheorem{thm}{Theorem}
\newtheorem{lem}{Lemma}

\title{Risk averse single machine scheduling - complexity and approximation}

\author{Adam Kasperski$^\dag$  Pawe{\l} Zieli{\'n}ski$^\ddag$\\        
	 {\small \textit{$^\dag$Department of Operations Research,
	 Faculty of Computer Science and Management,}}\\
	{\small \textit{Wroc{\l}aw University of Science and Technology,  Wroc{\l}aw, Poland}}\\
	 {\small \textit{$^\ddag$Department of Computer Science,
          Faculty of Fundamental Problems of Technology,}}\\
	{\small \textit{Wroc{\l}aw University of  Science and Technology,  Wroc{\l}aw, Poland}}
}

 \date{}

\begin{document}

\maketitle

\begin{abstract}
In this paper a class of single machine scheduling problems is considered. It is assumed that job processing times and due dates can be uncertain and they are specified in the form of discrete scenario set. A probability distribution in the scenario set is known. In order to choose a schedule some risk criteria such as the value at risk (VaR) an conditional value at risk (CVaR) are used. Various positive and negative complexity results are provided for basic single machine scheduling problems. In this paper new complexity results are shown and some known complexity results are strengthen.
 \end{abstract}

\section{Introduction}

Scheduling under risk and uncertainty has attracted considerable attention in recent literature. In practical applications of scheduling models the  exact values of input parameters, such as job processing times or due dates, are often unknown in advance. Hence, a solution must be computed, before the true realization of the input data reveals. Typically, a \emph{scenario set} $\mathcal{U}$ is a part of the input, which contains all possible realizations of the problem parameters, called \emph{scenarios}. If the probability distribution in $\mathcal{U}$ is unknown, then \emph{robust optimization} framework can be applied and solution performance in a worst case is optimized. First robust scheduling problems have been discussed in~\cite{DK95, KY97, YK93}. Two uncertainty representations, namely a \emph{discrete} and \emph{interval} ones were considered. In the former, scenario set $\mathcal{U}$ contains a finite number of distinct scenarios. In the latter, for each uncertain parameter an interval of its possible values is specified and $\mathcal{U}$ is the Cartesian product of these intervals. In order to compute a solution the minmax and minmax regret criteria can be applied. Minmax (regret) scheduling problems have various complexity properties, depending on the cost function and the uncertainty representation (see, e.g.,~\cite{AV00, AL06, K05, AAK11, DR16}). For a survey of  minmax (regret) scheduling problems
we refer the reader to~\cite{KZ14s}.

The robust scheduling models have well known drawbacks. Minimizing the maximum cost can lead to very conservative solutions. The reason is that the probability of occurrence of the worst scenario may be very small and the information connected with the remaining scenarios is ignored while computing a solution. One method of overcoming this drawback was given in~\cite{KZ16d}, where the OWA criterion, proposed in~\cite{YA88},
 was applied to compute an optimal schedule. In this approach, a set of weights is specified by the decision maker, which reflect his attitude towards a risk. The OWA operator contains the maximum, average and Hurwicz criteria as special cases. However, it does not take into account  a probabilistic information, which may be available for scenario set~$\mathcal{U}$.

In the case, when a probability distribution in $\mathcal{U}$ is known, the stochastic scheduling models are considered. The parameters of scheduling problem are then random variables with known probability distributions. Under this assumption, the expected solution performance is typically optimized  
(see, e.g.,~\cite{MSU99,PI08,SU05,SSU16}). However, this criterion assumes that the decision maker is risk neutral and leads to solutions that guarantee an optimal long run performance. Such a solution may be questionable, for example, if it is implemented only once (see, e.g.,~\cite{KY97}). 
In this case, the decision maker attitude towards a risk should be taken into account.

In~\cite{KPU02} a criterion called \emph{conditional value at risk} (CVaR) was applied to a stochastic portfolio selection problem. Using this criterion, the decision maker provides a parameter $\alpha\in [0,1)$, which reflects his attitude towards a risk. When $\alpha=0$, then CVaR becomes the expectation. However, for greater value of $\alpha$, more attention is paid to the worst outcomes, which fits into the robust optimization framework. The conditional value at risk is closely connected with the \emph{value at risk} (VaR) criterion
(see, e.g.,~\cite{P00}), which is just the $\alpha$-quantile of a random outcome. Both risk criteria have attracted considerable attention in stochastic optimization~(see, e.g.,~\cite{NST15, ZSZD17, NIK10, O12}). This paper is motivated by the recent papers~\cite{SSL14} and~\cite{NBN17}, in which the following stochastic scheduling models were discussed. We are given a scheduling problem with discrete scenario set $\mathcal{U}$. Each scenario $\xi_i\in \mathcal{U}$ is a realization of the problem parameters (for example, processing times and due dates), which can occur with a known positive probability ${\rm Pr}[\xi_i]$.  The cost of a given schedule is a discrete random variable with the probability distribution induced by the probability distribution in $\mathcal{U}$. The VaR and CVaR criteria, with a fixed level $\alpha$, are used to compute a best solution. 

In~\cite{SSL14} and~\cite{NBN17} solution methods based on mixed integer programming models were proposed to minimize VaR and CVaR in scheduling problems with the total weighted tardiness criterion.
The aim of this paper is to analyze the models discussed  in~\cite{SSL14} and~\cite{NBN17}  from the complexity point of view. We will consider the class of single machine scheduling problems with basic cost functions, such as the maximum tardiness, the total flow time, the total tardiness and the number of late jobs. We will discuss also the weighted versions of these cost functions. We provide a picture of computational complexity for all these problems by proving some positive and negative complexity results. Since VaR and CVaR  generalize the maximum criterion, we can use some results known from robust minmax scheduling. The complexity results for minmax versions of single machine scheduling problems under discrete scenario set were obtained 
in~\cite{AAK11, AC08,DK95, MNO13}. In this paper we will show that some of these results can be strengthen.

This paper is organized as follows. In Section~\ref{sec1} we recall the definitions of the VaR and CVaR criteria and show their properties, which will be used later on. In Section~\ref{sec2} the problems discussed in this paper are defined. In Section~\ref{sec3} some general relationships between the problems with various risk criteria are shown. Finally, Sections~\ref{sec4} and~\ref{sec5} contain some new negative and positive complexity results for the the considered problems. These results are summarized in the tables presented in Section~\ref{sec2}.


\section{The risk criteria}
\label{sec1}

Let $\rm Y$ be a random variable. We  will consider the following risk criteria~\cite{P00,RU00}:
  \begin{itemize}
\item \emph{Value at Risk} ($\alpha$-quantile of $\mathrm{Y}$): $$
{\bf VaR}_\alpha[\rm Y]=\inf\{t: {\rm Pr}[Y\leq t] \geq \alpha\}, \alpha\in(0,1],$$
\item \emph{Conditional value at risk}:
$${\bf CVaR}_\alpha[\rm Y]=\inf\{\gamma+\frac{1}{1-\alpha}{\bf E}[\rm Y-\gamma]^+: \gamma \in \Rset\}, \alpha\in [0,1),$$
\end{itemize}
where $[x]^+=\max\{0,x\}$.
Assume that $\mathrm{Y}$ is a discrete random variable taking nonnegative values $b_1,\dots, b_K$. Then $ {\bf VaR}_\alpha[\rm Y]$ and ${\bf CVaR}_\alpha[\rm Y]$ can be computed by using the following programs, respectively (see, e.g.,~\cite{NBN17,O12,RU00}):
\begin{align}
		\text{(a)}\;\;\;\;\;& \min  \theta & \text{(b)}\;\;\;\;\;&\min   \gamma + \frac{1}{1-\alpha}\sum_{i\in [K]} {\rm Pr}[\mathrm{Y}=b_k] u_k\nonumber\\
		\text{s.t. }& b_k-\theta\leq M \beta_k,\;\;\; k\in [K] &\text{s.t. }& \gamma + u_k \geq  \displaystyle b_k,\;\;\; k\in [K] \label{exVC1}\\
		 &\sum_{k\in [K]} {\rm Pr}[\mathrm{Y}=b_k]  \beta_k \leq 1-\alpha &&  u_k\geq 0,\;\;\;  k\in [K]\nonumber\\ 
		 & \beta_k\in \{0,1\},\;\;\;  k\in [K] &&	\nonumber 	
\end{align}
where $M\geq \max\{b_1,\dots, b_K\}$ and $[K]=\{1,\dots,K\}$.  Notice that (\ref{exVC1})b is a linear programming problem.  In the following, we will use the following dual to~(\ref{exVC1})b:
\begin{equation}
\label{cvarmod1}
	\begin{array}{lll}
		\displaystyle \max  & \displaystyle \sum_{k\in [K]} b_k r_k  \\
					\text{s.t.} & \displaystyle\sum_{k\in [K]} r_k=1\\
						       &0 \leq r_k\leq \frac{{\rm Pr}[\mathrm{Y}=b_k] }{1-\alpha}, & k\in [K]
	\end{array}	
\end{equation}
The equality constraint in~(\ref{cvarmod1}) follows from the fact that all $b_k$, $k\in [K]$, are nonnegative.
Substituting $r_k=q_k/(1-\alpha)$ into~(\ref{cvarmod1}), we get the following equivalent formulation for ${\bf CVaR}_\alpha[\rm Y]$:
\begin{equation}
\label{cvarmod2}
	\begin{array}{lll}
		\displaystyle \max  & \displaystyle \frac{1}{1-\alpha} \sum_{k\in [K]} b_k   q_k\\
					\text{s.t.} & \displaystyle\sum_{k\in [K]} q_k=1-\alpha\\
						       &0 \leq q_k\leq  {\rm Pr}[\mathrm{Y}=b_k],  & k\in [K]
	\end{array}	
\end{equation}
Program~(\ref{cvarmod2}) can be solved by using a greedy method, which is illustrated in Figure~\ref{fig0}. Namely, we fix the optimal values of $q_k$ by greedily distributing the amount $1-\alpha$ among the largest values of $b_i$.
It is easy to see that ${\bf CVaR}_{0}[{\rm Y}]={\bf E}[{\rm Y}]=\sum_{k\in [K]} b_k {\rm Pr}[\mathrm{Y}=b_k]$. On the other hand, ${\bf CVaR}_{1-\epsilon}[{\rm Y}]= {\bf VaR}_1[{\rm Y}]={\bf Max}[{\rm Y}]=\max_{k\in [K]} b_k$ for sufficiently small $\epsilon>0$ and any probability distribution.

\begin{figure}[ht]
	\centering
	\includegraphics[width=6cm]{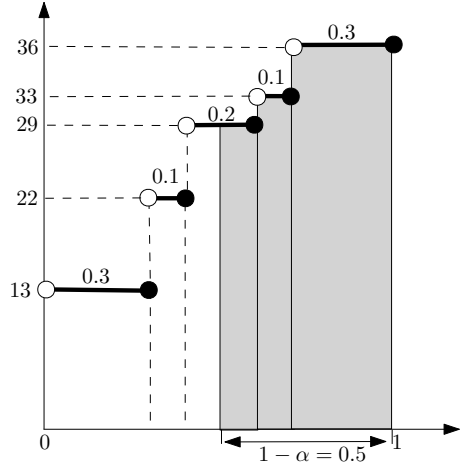}
	\caption{A computation of ${\bf CVaR}_{0.5}[{\rm Y}]$ for $\rm Y$ taking the values of 13, 22, 29, 33, and 36 with the probabilities 0.3, 0.1, 0.2, 0.1, 0.3, respectively. The value of ${\bf CVaR}_{0.5}[{\rm Y}]$ is the grey area  divided by $1-\alpha=0.5$.} \label{fig0}
\end{figure}

We now show several properties of the  risk measures which will be used later on in this paper.
 \begin{lem}
\label{lem01}
	Let $\mathrm{Y}$ be a discrete random variable which can take $K$ nonnegative values $b_1,\dots, b_K$. 
	The following inequalities hold for each $\alpha\in [0,1)$:
	\begin{equation}
		{\bf E}[\mathrm{Y}]\leq {\bf CVaR}_\alpha [\mathrm{Y}] \leq
		 \min\left \{\frac{1}{{\Pr}_{\min}},\frac{1}{1-\alpha} \right\} {\bf E}[\mathrm{Y}],
	\end{equation}
	where ${\rm Pr}_{\min}=\min_{k\in [K]} {\rm Pr}[\mathrm{Y}=b_k]$.
\end{lem}
\begin{proof}
	Fix $\alpha\in [0,1)$.
	The inequality ${\bf E}[\mathrm{Y}]\leq {\bf CVaR}_\alpha [\mathrm{Y}]$ follows directly from the definition of the expected value and the conditional value at risk. We now prove the second inequality. Let $r^*_1, \dots r^*_k$ be the optimal values in~(\ref{cvarmod1}). Then the inequality
$${\bf CVaR}_\alpha [\mathrm{Y}]=\sum_{k\in [K]} r^*_k b_k \leq \sum_{k\in [K]} \frac{{\rm Pr}[\mathrm{Y}=b_k]}{(1-\alpha)} b_k= \frac{1}{1-\alpha}{\bf E}[\mathrm{Y}]$$ 
holds. Since the value of ${\bf CVaR}_\alpha [\mathrm{Y}]$ 
is a convex combination of $b_1,\ldots,b_k$ (see (\ref{cvarmod1})), we have
${\bf CVaR}_\alpha [\mathrm{Y}] \leq {\bf Max}[\mathrm{Y}]=b_{\max}\leq \sum_{k\in [K]} \frac{{\rm Pr}[\mathrm{Y}=b_k]}{{\rm Pr}_{\min}} b_k=\frac{1}{{\rm Pr}_{\min}}{\bf E}[\mathrm{Y}],$
and the lemma follows.
\end{proof}

\begin{lem}
\label{lemowa1}
		Let $\mathrm{X}$ and $\mathrm{Y}$ be two discrete random variables taking nonnegative values
		 $a_1,\dots, a_K$, and $b_1,\dots, b_K$, respectively, with 
		${\rm Pr}[\mathrm{X}=a_i]={\rm Pr}[\mathrm{Y}=b_i]$ and $a_i\leq \gamma b_i$ for each $i\in [K]$ and some fixed $\gamma \geq 0$. Then ${\bf CVaR}_\alpha[\mathrm{X}]\leq \gamma {\bf CVaR}_\alpha[\mathrm{Y}]$ for each $\alpha\in [0,1)$ and ${\bf VaR}_\alpha[\mathrm{X}]\leq \gamma {\bf VaR}_\alpha[\mathrm{Y}]$ for each $\alpha\in (0,1]$.
\end{lem}
\begin{proof}
Let us compute ${\bf CVaR}_\alpha[\mathrm{X}]$ by using~(\ref{cvarmod1}) and denote by $r^*_k$, $k\in [K]$, the optimal values in~(\ref{cvarmod1}). Then ${\bf CVaR}_\alpha[\mathrm{X}]=\sum_{k\in [K]} r^*_k a_k\leq \gamma \sum_{k\in [K]} r^*_k b_k\leq \gamma {\bf CVaR}_\alpha[\mathrm{Y}]$.
Let us compute ${\bf VaR}_{\alpha}[\mathrm{Y}]$ by solving the problem~(\ref{exVC1})a.
 Let $\theta^*$, $\beta^*_k$, $k\in [K]$,  be an optimal solution to~(\ref{exVC1})a. Since $\gamma\geq 0$, the constraint
$\gamma b_k -\gamma \theta^* \leq \gamma M \beta^*_k$ holds
for each $k\in [K]$. By $a_k\leq \gamma b_k$ for each $k\in [K]$, we get 
$a_k-\gamma \theta^* \leq M' \beta^*_k,$
where $M'=\gamma M \geq \max\{a_1,\dots, a_K\}$, $k\in [K]$. In consequence, 
\begin{equation}
\label{exVC12}
	\begin{array}{llll}
		 & \displaystyle a_k-\gamma \theta^*\leq M' \beta^*_k & k\in [K]\\
		 &\displaystyle  \sum_{k\in [K]} {\rm Pr}[\mathrm{X}=a_k] \cdot  \beta^*_k \leq 1-\alpha\\
	\end{array}
\end{equation}
and ${\bf VaR}_{\alpha}[\mathrm{X}]\leq \gamma \theta^*=\gamma{\bf VaR}_{\alpha}[\mathrm{Y}]$.
\end{proof}

\section{Problem formulations}
\label{sec2}

We are given a set $J$ of $n$ jobs, which can be partially ordered by some precedence constraints. Namely, $i\rightarrow j$ means that job $j$ cannot start before job $i$ is completed.  For each job $j\in J$ a nonnegative processing time $p_j$, a nonnegative due date $d_j$ and a nonnegative weight $w_j$ can be specified.  A schedule $\pi$ is a feasible (i.e. preserving the precedence constraints) permutation of the jobs and $\Pi$ is the set of all feasible schedules. We will use $C_j(\pi)$ to denote the completion time of job $j$ in schedule $\pi$. 
Obeying the standard notation, we will use $T_j(\pi)=[C_j(\pi)-d_j]^+$ to define the tardiness of $j$ in $\pi$, and $U_j(\pi)=1$ if $C_j(\pi)>d_j$ (job $j$ is late in $\pi$) and $U_j(\pi)=0$ (job $j$ is on-time in $\pi$), otherwise. In the deterministic case we seek a schedule $\pi\in \Pi$ that minimizes a given  cost function~$f(\pi)$. 
The basic cost functions are the \emph{total flow time}~$\sum_{j\in J} C_j(\pi)$, the 
\emph{total tardiness}~$\sum_{j\in J} T_j(\pi)$, the \emph{maximum tardiness}~$\max_{j\in J} T_j(\pi)$ and the 
\emph{total number of late jobs}~ $\sum_{j\in J} U_j(\pi)$. We can also consider the weighted versions of these functions.  Scheduling problems $\mathcal{P}$ will be denoted by means of  the standard Graham's notation (see, e.g.,~\cite{B07}).

In this paper we assume that job processing times and due dates can be uncertain. 
The uncertainty is modeled by
a discrete \emph{scenario set} $\mathcal{U}=\{\xi_1,\xi_2,\dots,\xi_K\}$.  Each realization of the parameters $\xi \in \mathcal{U}$ is called a \emph{scenario}. For each scenario $\xi\in \mathcal{U}$ a probability ${\rm Pr}[\xi]$ of its occurrence is known (without loss of generality we can assume~ ${\rm Pr}[\xi]>0$).
We will use $p_j(\xi)$ and $d_j(\xi)$  to denote  the processing time and due date of job~$j$ 
under scenario $\xi\in \mathcal{U}$, respectively. We will  denote by $C_j(\pi, \xi)$, $T_j(\pi, \xi)$ and $U_j(\pi, \xi)$ the completion time, tardiness and unit penalty of job $\pi$, respectively, under scenario $\xi\in \mathcal{U}$.  Also, $f(\pi, \xi)$ stands for the cost of schedule $\pi$ under scenario $\xi\in \mathcal{U}$.
 Given a feasible schedule $\pi \in \Pi$, we denote by ${\rm F}(\pi)$ a random cost of $\pi$. Notice that ${\rm F}(\pi)$ is a discrete random variable with the probability distribution induced by the probability distribution in~$\mathcal{U}$.
   
For a fixed value of $\alpha$, we can compute a performance measure of $\pi$, namely the expected cost $\textbf{E}[{\rm F}(\pi)]$, the maximum cost $\textbf{Max}[{\rm F}(\pi)]$, the value at risk $\textbf{VaR}_\alpha[{\rm F}(\pi)]$ and the conditional value at risk $\textbf{CVaR}_\alpha[{\rm F}(\pi)]$. A sample problem $1||\sum C_j$ with 4 jobs and 5 processing time scenarios is shown in Figure~\ref{fig1}. Let $\pi=(1,2,3,4)$. It is easily seen that $\textbf{E}[{\rm F}(\pi)]=26$, $\textbf{VaR}_{0.5}[{\rm F}(\pi)]=29$,  $\textbf{CVaR}_{0.5}[{\rm F}(\pi)]=34$ and $\textbf{Max}[{\rm F}(\pi)]=36$.

 \begin{figure}[ht]
	\centering
	\includegraphics[width=12cm]{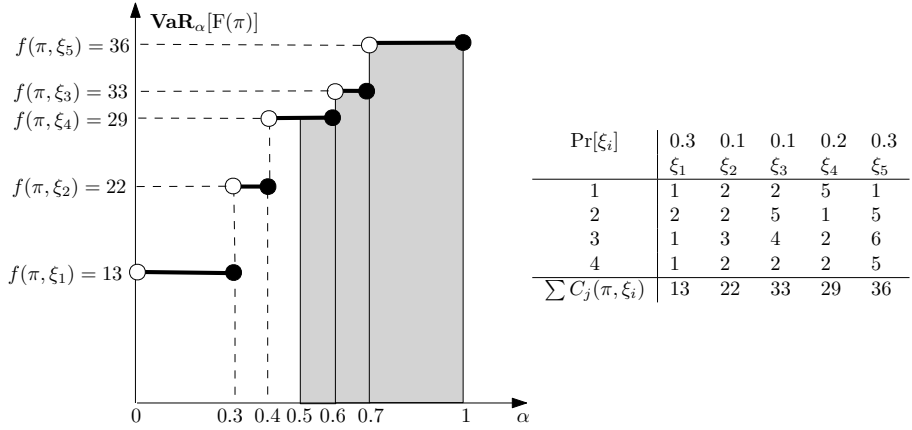}
	\caption{A sample scheduling problem $1||\sum C_j$ with 5 processing time scenarios.} \label{fig1}
\end{figure}
 
 In this paper we will study the problems $\textsc{Min-VaR}_\alpha~\mathcal{P}$, $\textsc{Min-CVaR}_\alpha~\mathcal{P}$, $\textsc{Min-Exp}$~$\mathcal{P}$, and  $\textsc{Min-Max}~\mathcal{P}$, in which we minimize the corresponding performance measure for a fixed $\alpha$ and a specific
 single machine scheduling problem~$\mathcal{P}$, under a given scenario set $\mathcal{U}$. Notice that the robust $\textsc{Min-Max}~\mathcal{P}$ problem is a special case of both  $\textsc{Min-VaR}_\alpha~\mathcal{P}$ and $\textsc{Min-CVaR}_\alpha~\mathcal{P}$. Also, $\textsc{Min-Exp}$~$\mathcal{P}$ is a special case of $\textsc{Min-CVaR}_\alpha~\mathcal{P}$.

In the next sections we provide a number of new positive and negative complexity and approximation results for basic single machine scheduling problems $\mathcal{P}$.
 Tables~\ref{tabs1}-\ref{tabs3} summarize 
 the known and new results. 
In Table~\ref{tabs1}, the negative results for uncertain due dates and deterministic processing times are shown. In Table~\ref{tabs2}, the negative results for uncertain processing times and deterministic due dates are presented. Finally, in Table~\ref{tabs3}, 
some positive results are shown.
 \begin{table}[ht]
\centering
\footnotesize
\caption{Complexity results for uncertain due dates (processing times are deterministic).} \label{tabs1}
	\begin{tabular}{l|llllllllll}
	$\mathcal{P}$	& $\textsc{Min-Exp}~\mathcal{P}$ & $\textsc{Min-VaR}_\alpha~\mathcal{P}$ & $\textsc{Min-CVaR}_\alpha~\mathcal{P}$ & $\textsc{Min-Max}~\mathcal{P}$ \\ \hline
$1|p_j=1|T_{\max}$ & str. NP-hard &  str. NP-hard  & str. NP-hard & poly. sol.~\cite{KZ16d}\\
		               & not appr. within & not at all appr. & not appr. within\\
		               & $\frac{7}{6}-\epsilon$, $\epsilon>0$~\cite{KZ16d}	 & for any $\alpha\in (0,1)$ & $\frac{7}{6}-\epsilon$, $\epsilon>0$\\
		               & 							 &					  & for any $\alpha\in (0,1)$ \\ \hline
		      
$1|p_j=1|\sum U_j$  & poly sol. & str. NP-hard & str. NP-hard & str. NP-hard \\
		       &	(assignment)	     & not at all appr. & for any $\alpha\in (0,1)$ & not appr. for any \\
		       &		     & for any $\alpha\in (0,1)$ & & constant $\gamma>1$ \\ \hline
$1||\sum U_j$ & NP-hard & as above & as above & as above \\ \hline		      		       
$1|p_j=1|\sum T_j$  & poly sol. & str. NP-hard & str. NP-hard & str. NP-hard \\
			& (assignment)		& not at all appr. & for any $\alpha\in (0,1)$		& not appr. within \\
			&		& for any $\alpha\in (0,1)$  & & $\frac{5}{4}-\epsilon$, $\epsilon>0$ \\ \hline
$1||\sum T_j$ & str. NP-hard & as above & as above & as above \\
\hline
	\end{tabular}
\end{table}
 \begin{table}[ht]
\centering
\footnotesize
\caption{Complexity results for uncertain processing times (the due dates are deterministic).} \label{tabs2}
	\begin{tabular}{l|llllllllll}
$\mathcal{P}$		& $\textsc{Min-Exp}~\mathcal{P}$ & $\textsc{Min-VaR}_\alpha~\mathcal{P}$ & $\textsc{Min-CVaR}_\alpha~\mathcal{P}$ & $\textsc{Min-Max}~\mathcal{P}$ \\ \hline
$1||\sum C_j$ &    poly sol. &  str. NP-hard  & str. NP-hard & str. NP-hard\\
		      &  	      & not appr. within & for any $\alpha\in (0,1)$ & not appr. within \\
		      &		      & $\frac{6}{5}-\epsilon$, $\epsilon>0$ & & $\frac{6}{5}-\epsilon$, $\epsilon>0$~\cite{KY97, MNO13}\\ \hline
$1||\sum U_j$  & open             & str. NP-hard & str. NP-hard & str. NP-hard \\
		       & 		     & for any $\alpha\in (0,1)$  &  for any $\alpha\in (0,1)$ & \\ \hline
$1||\sum T_j$  & NP-hard~\cite{LA77} & str. NP-hard & str. NP-hard & str. NP-hard \\
			&		& not  appr. within &  for any $\alpha\in [0,1)$ & not appr. within  \\
			&		&$\frac{6}{5}-\epsilon$, $\epsilon>0$  &  & $\frac{6}{5}-\epsilon$, $\epsilon>0$\\
			\hline
	\end{tabular}
\begin{flushleft}	
\end{flushleft}	
\end{table}
\begin{table}[ht]
\centering
\footnotesize
	\caption{Positive complexity results.} \label{tabs3}
	\begin{tabular}{l|llllllllll}
$\mathcal{P}$		& $\textsc{Min-Exp}~\mathcal{P}$ & $\textsc{Min-VaR}_\alpha~\mathcal{P}$ & $\textsc{Min-CVaR}_\alpha~\mathcal{P}$ & $\textsc{Min-Max}~\mathcal{P}$ \\ \hline
$1|prec|\max w_j T_j$ & $O(f_{\max}^K Kn^2)$ &  $O(f_{\max}^K Kn^2)$ &  $O(f_{\max}^K Kn^2)$ & 
 $O(Kn^2)$~\cite{KZ16d}\\ 
                                    & FPTAS &  FPTAS&  FPTAS&\\ 
                                    &for const. $K$ &for const. $K$ &for const. $K$ &\\\hline
$1|prec|\sum w_j C_j$ &  as the determ. &  appr. within 2 &appr. within 2 & appr. within 2~\cite{MNO13}\\
		      &	 problem & for const. $K$     & \\ \hline
$1|prec^*|\sum w_j C_j$ &  poly sol. &  appr. within 2 &appr. within  & appr. within 2~\cite{MNO13}\\ 
					&	      &  for const. $K$         & $\min\{\frac{1}{1-\alpha},2\}$ \\ \hline
$1|p_j=1|\sum w_j U_j$  &   poly sol. & - & appr. within & appr. within $K$ \\
		       &		     &   &  $\min\{\frac{1}{{\rm Pr}_{\min}},\frac{1}{1-\alpha}\}$ &  \\ \hline
$1||\sum w_j U_j$  & appr. within            & - & appr. within & appr. within  \\
determ. proc. times		       & $4+\epsilon$, $\epsilon>0$ &   &  $\min\{\frac{4+\epsilon}{{\rm Pr}_{\min}},\frac{4+\epsilon}{1-\alpha}\}$ & $(4+\epsilon)K$, $\epsilon>0$ 
		        \\ \hline		       
$1|p_j=1|\sum w_j T_j$  & poly sol. & - &   appr. within & appr. within $K$ \\
			&		&  &  $\min\{\frac{1}{{\rm Pr}_{\min}},\frac{1}{1-\alpha}\}$ & \\ \hline
$1||\sum w_j T_j$  & appr. within            & - & appr. within & appr. within \\
determ. proc. times		       & $4+\epsilon$, $\epsilon>0$ &   &  $\min\{\frac{4+\epsilon}{{\rm Pr}_{\min}},\frac{4+\epsilon}{1-\alpha}\}$ & $(4+\epsilon)K$, $\epsilon>0$ 
		         \\ \hline				
	\end{tabular}
\begin{flushleft}	
	{\footnotesize $f_{\max}$ is an upper bound
	on the cost of any schedule under any scenario; $prec^*$ is a polynomially solvable structure of the precedence constraints; ${\rm Pr}_{\min}=\min_{k\in [K]} {\rm Pr}[\xi_k]$.}
\end{flushleft}	
\end{table}

\section{Some general properties}
 \label{sec3}
 
 In this section we will show some general relationships between the problems with various performance criteria. These properties will be used later to establish some positive and negative complexity results for particular problems.
 
\begin{thm}
	The following statements are true:
	\begin{enumerate}
	\item If $\textsc{Min-Exp}~\mathcal{P}$ is approximable within $\sigma>1$
	(for $\sigma=1$ it is
	 polynomially solvable), then $\textsc{Min-CVaR}_\alpha~\mathcal{P}$ is approximable within $\sigma\rho$,
	 where $\rho=
	 \min\{\frac{1}{{\rm Pr}_{\min}},\frac{1}{1-\alpha}\}$, for each constant $\alpha\in [0,1)$.
	\item If $\textsc{Min-Exp}~\mathcal{P}$ with $K$-scenarios is NP-hard and hard to approximate within $\rho>1$, then $\textsc{Min-CVaR}_\alpha~\mathcal{P}$ with $K+1$ scenarios is also NP-hard and hard to approximate within $\rho$ for each constant $\alpha\in [0,1)$.
	\end{enumerate}
	\label{thmexp}
\end{thm}
\begin{proof}
	We first prove assertion~1.
	Let $\pi^*$ minimize the expected cost and $\pi'$  minimize the conditional value at risk for a fixed $\alpha \in [0,1)$. 
	We will denote by $\hat{\pi}$ a $ \sigma$-approximation schedule for $\textsc{Min-Exp}~\mathcal{P}$.
	Using Lemma~\ref{lem01} we get
	$${\bf CVaR}_\alpha [{\rm F}(\hat{\pi})]\leq 
	\rho  {\bf E}[{\rm F}(\hat{\pi})]\leq \sigma\rho  {\bf E}[{\rm F}(\pi^*)]
	 \leq   \sigma\rho {\bf E}[{\rm F}(\pi')]\leq  \sigma\rho {\bf CVaR}_\alpha [{\rm F}(\pi')],$$
	and the  assertion follows.
	
	\begin{figure}[ht]
	\centering
	\includegraphics[width=6cm]{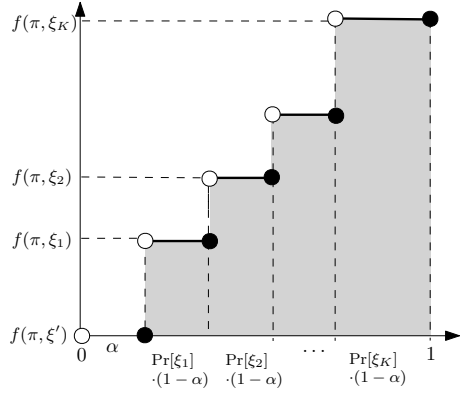}
	\caption{Illustration of the proof of Theorem~\ref{thmexp}.} \label{fig2a}
\end{figure}
	
	In order to prove assertion~2, consider an instance of $\textsc{Min-Exp}~\mathcal{P}$ with $\mathcal{U}=\{\xi_1,\dots,\xi_K\}$.
	 Fix $\alpha\in (0,1)$ (the statement trivially holds for $\alpha=0$) and add one additional scenario~$\xi'$ under which the cost of each schedule is~0 (for example, all job processing times are~0 under~$\xi'$). We fix  ${\rm Pr}'[\xi']=\alpha$ and ${\rm Pr}'[\xi_i]={\rm Pr}[\xi_i]\cdot(1-\alpha)$ for each $i\in [K]$. Denote by ${\rm F}'(\pi)$ the random cost of $\pi$ under the new scenario set~$\mathcal{U}'$.
	 For each schedule $\pi$ we get  (see Figure~\ref{fig2a}): 
	 $$ {\bf CVaR}_{\alpha}[{\rm F}'(\pi)]=\frac{1}{1-\alpha}\sum_{i\in [K]} {\rm Pr}'[\xi_i] f(\pi, \xi_i)=\sum_{i\in [K]}{\rm Pr}[\xi_i] f(\pi, \xi_i)={\bf E}[{\rm F}(\pi)].$$
	 Hence there is a cost preserving reduction from $\textsc{Min-Exp}~\mathcal{P}$ with $K$ scenarios to $\textsc{Min-CVaR}_\alpha~\mathcal{P}$ with $K+1$ scenarios and the theorem follows.
\end{proof}
 \begin{thm}
\label{thm01}
	Assume that $w_j=1$ for each job $j\in J$ in problem $\mathcal{P}$.
	If \textsc{Min-Max}~$\mathcal{P}$ with $K\geq 2$ scenarios is NP-hard and hard to approximate within $\rho>1$, then
	\begin{enumerate}	
	\item $\textsc{Min-VaR}_\alpha~\mathcal{P}$ with $K+1$ scenarios is NP-hard  and hard to approximate within 
	$\rho>1$ for each constant $\alpha\in (0,1)$.	
	\item $\textsc{Min-CVaR}_\alpha~\mathcal{P}$ with $K+1$ scenarios is NP-hard  for  each constant $\alpha\in (0,1)$.	 
	 \end{enumerate}
\end{thm}\begin{proof}
	Choose an instance of the $\textsc{Min-Max}~\mathcal{P}$ problem with $\mathcal{U}=\{\xi_1,\dots, \xi_K\}$, $K\geq 2$. Fix $\alpha\in (0,1)$ and create $\mathcal{U}'$ by adding to $\mathcal{U}$  a dummy scenario $\xi'$ such that  the cost of each schedule under $\xi'$ equals $M$ and $M\geq f(\pi, \xi_i)$ for each $i\in [K]$ and each $\pi \in \Pi$.  It is enough to fix $p_j(\xi')=p_{\max}$ and $d_j(\xi')=d_{\min}$ for each job $j\in J$, where $p_{\max}=\max_{j\in J, i\in [K]} p_j(\xi_i)$ is the maximum job processing time and $d_{\min}=\min_{j\in J, i\in [K]} d_j(\xi_i)$ is the minimum due date over all scenarios. 
	 For each of the two assertions, we define an appropriate probability distribution in $\mathcal{U}'$. We will use ${\rm F}'(\pi)$ to denote the random cost of $\pi$ under $\mathcal{U}'$.
	 
	 In order to prove the statement~1, we fix ${\rm Pr}[\xi']=1-\alpha$ and ${\rm Pr}[\xi_i]=\frac{\alpha}{K}$ for each $i\in [K]$ (see Figure~\ref{fig2}a). The equality
	 ${\bf VaR}_\alpha[{\rm F}'(\pi)]={\bf Max}[{\rm F}(\pi)]$ holds.
	Hence, there is a cost preserving reduction from $\textsc{Min-Max}~\mathcal{P}$ with $K$ scenarios to $\textsc{Min-VaR}_\alpha~\mathcal{P}$ with $K+1$ scenarios and the statement follows.	
	To prove the statement~2,  we fix ${\rm Pr}[\xi']=\gamma$ and ${\rm Pr}[\xi_i]=\beta$ for each $i\in [K]$, where $\gamma$ and $\beta$ satisfy the following system of equations (see Figure~\ref{fig2}b):
			
			$$ \left\{
			\begin{array}{ll}
				\beta+\gamma=1-\alpha \\
				K\beta+\gamma=1
			\end{array}\right.$$
	In consequence $\beta=\frac{\alpha}{K-1}$ and $\gamma=1-\frac{K\alpha}{K-1}$. Observe that $\beta>1$ as $\alpha\in (0,1)$.
	\begin{figure}[ht]
	\centering
	\includegraphics[width=14cm]{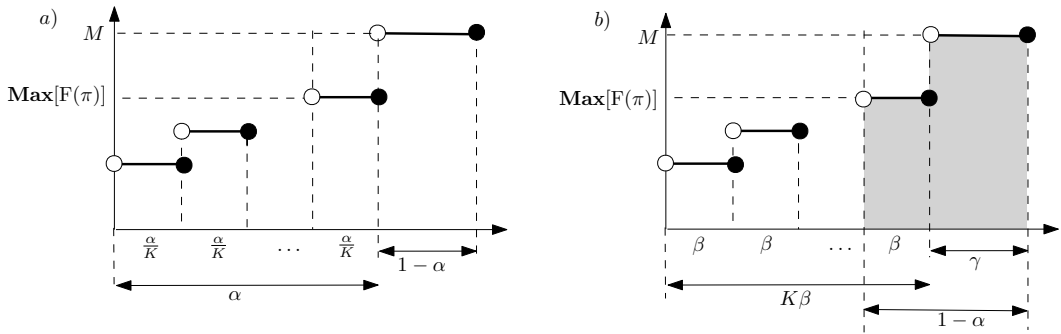}
	\caption{Illustration of the proof of Theorem~\ref{thm01}.} \label{fig2}
\end{figure}
	 For each schedule $\pi$ we get
	$${\bf CVaR}_{\alpha}[{\rm F}'(\pi)]=\frac{1}{1-\alpha}(\beta \cdot {\bf Max}[\mathrm{F}(\pi)] + \gamma M),$$
	where $\beta, \gamma$ and $M$ are numbers depending on $K$ and $\alpha$.
	Hence $\textsc{Min-Max}~\mathcal{P}$ and the corresponding instance of $\textsc{Min-CVaR}_\alpha~\mathcal{P}$ have the same optimal solutions and the theorem follows.
\end{proof}	

\section{Negative complexity results}
\label{sec4}

In this section we will prove some negative complexity results for basic single machine scheduling problems.  These results are summarized in Tables~\ref{tabs1} and~\ref{tabs2}. 

\subsection{Uncertain due dates}

We first address  the problem of minimizing the value at risk criterion. The following theorem characterizes the complexity of some basic problems:
\begin{thm} 
\label{thm1}
	For each  $\alpha\in (0,1)$, $\textsc{Min-Var}_\alpha~\mathcal{P}$ is strongly NP-hard and not at all approximable, when $\mathcal{P} \in \{1|p_j=1|T_{\max},\,
	1|p_j=1|\sum T_j, \,1|p_j=1| \sum U_j\}$.
\end{thm}
\begin{proof}
	Consider an instance of the following strongly NP-hard \textsc{Min 3-Sat} problem~\cite{KM94, AZ02}. 
	We are given boolean variables $x_1,\dots, x_n$, a collection of clauses $\mathcal{C}_1,\dots \mathcal{C}_m$, where each clause is a disjunction of at most $3$ literals (variables or their negations) and we ask if there is an assignment to the variables which satisfies at most $L < m$
	 clauses. We can ask equivalently, if there is an assignment to the variables for which at least $l=m-L$ clauses are not satisfied.
		
	Given an instance of \textsc{Min 3-Sat}, we create two jobs $J_{x_i}$ and $J_{\overline{x}_i}$ for each variable $x_i$, 
	$i\in [n]$.
	A \emph{due date scenario}~$\xi_i$ corresponds to clause $\mathcal{C}_i=(l_1 \vee l_2 \vee l_3)$ and is formed as follows. For each $q=1,2,3$, if $l_q=x_j$, then the due date of $J_{x_j}$ is $2j-1$ and the due date of $J_{\overline{x}_j}$ is $2j$; if $l_q=\overline{x}_j$, then the due date of $J_{x_j}$ is $2j$ and the due date of $J_{\overline{x}_j}$ is $2j-1$; if neither $x_j$ nor $\overline{x}_j$ appears in $\mathcal{C}_i$, then the due dates of $J_{x_j}$ and $J_{\overline{x}_j}$ are set to $2j$.
	  An example is shown in Table~\ref{tab1}.

	\begin{table}[ht]
	  \centering
	  \caption{The set of jobs and the due date scenarios for the formula $(x_1\vee \overline{x}_2 \vee \overline{x}_3)\wedge (\overline{x}_2 \vee \overline{x}_3 \vee x_4) \wedge (\overline{x}_1 \vee  x_2 \vee \overline{x}_4) \wedge (x_1 \vee x_2 \vee x_3) \wedge (x_1 \vee x_3 \vee \overline{x}_4)$.} \label{tab1}
			\begin{tabular}{l|lllll}
										 & $\xi_1$ & $\xi_2$ & $\xi_3$ & $\xi_4$ & $\xi_5$ \\ \hline
					$J_{x_1}$ & 1  & 2 & 2 & 1 & 1\\
					$J_{\overline{x}_1}$  & 2 & 2 & 1 & 2 & 2 \\  \hline
					$J_{x_2}$ &  4 & 4 & 3 & 3 & 4\\
					$J_{\overline{x}_2}$ & 3 & 3 & 4 & 4 & 4 \\  \hline
					$J_{x_3}$ & 6 & 6 & 6 & 5 & 5\\
					$J_{\overline{x}_3}$ & 5 & 5 & 6 & 6 & 6 \\  \hline
					$J_{x_4}$ &  8 & 7 & 8 & 8 & 8\\
					$J_{\overline{x}_4}$ & 8 & 8 & 7 & 8 & 7 \\  \hline
			\end{tabular}
		\end{table}
		Let us define a subset of the schedules $\Pi'\subseteq \Pi$ such that each schedule $\pi\in \Pi'$ is of the form $\pi=(J_1,J_1',J_2,J_2',\dots,J_n,J_n')$, where $J_j,J_j'\in\{J_{x_j},J_{\overline{x}_j}\}$ for $j\in [n]$. Observe that $\Pi'$ contains exactly $2^n$ schedules and each such a schedule corresponds to the assignment to the variables such that $x_j=0$ if $J_{x_j}$ is processed before $J_{\overline{x}_j}$ and $x_j=1$ otherwise.
	 Note that this correspondence is one-to-one. In the following we assume that $f(\pi, \xi_i)$ is the maximum tardiness, or the total tardiness, or the sum of unit penalties in $\pi$ under $\xi_i$. The reasoning will be the same for each of these cost functions. 
If $\pi \notin \Pi'$, then $f(\pi, \xi_i)>0$ for each scenario $\xi_i$.
Indeed, suppose that $\pi \notin \Pi'$ and let $J_j$ ($J_j')$ be the last job in $\pi$ which is not placed properly, i.e. $J_j,(J_j')\notin\{J_{x_j},J_{\overline{x}_j}\}$. Then $J_j$ ($J_j'$) is late under all scenarios.	
On the other hand, 
			if $\pi \in \Pi'$, then the number of scenarios under which no  job is late is equal to the number of unsatisfiable clauses for the assignment corresponding to $\pi$.		
Fix $\alpha\in (0,1)$. We will add to $\mathcal{U}$ one additional scenario $\xi'$ and define a probability distribution in $\mathcal{U}$, depending on the fixed $\alpha$, so that
	 the answer to \textsc{Min 3-Sat} is yes if and only if there is schedule $\pi$ for which $\mathbf{VaR}_{\alpha}[\mathrm{F}(\pi)]\leq 0$.  This will prove the stated result. We consider two cases:
	 \begin{enumerate}
\item  $l/m\geq \alpha$. We create dummy scenario $\xi'$ under which the due date of all jobs is equal to~0. The probability of this scenario is equal to  $\frac{l-\alpha m}{l}$. The probability of each of the remaining scenarios is equal to $\frac{1}{m}(1-\frac{l-\alpha m}{l})=\frac{\alpha}{l}$. Assume that the answer to \textsc{Min 3-Sat} is yes. So, there is an assignment to the variables which satisfies at most $m-l$ clauses. By the above construction, there is a schedule~$\pi\in \Pi'$ whose cost is positive under at most $m-l$ scenarios plus the dummy one. It holds
		$$\mathrm{Pr}[\mathrm{F}(\pi)>0]\leq \frac{l-\alpha m}{l} + (m-l)\frac{\alpha}{l}=1-\alpha.$$
	Hence $\mathrm{Pr}[\mathrm{F}(\pi)\leq 0]\geq \alpha$ and $\mathbf{VaR}_{\alpha}[\mathrm{F}(\pi)]\leq 0$. Assume that the answer to \textsc{Min 3-Sat} is no. Then, for every schedule $\pi$ there are more than $m-l$ scenarios under which the cost of $\pi$ is positive plus the dummy one. Hence 
$\mathrm{Pr}[\mathrm{F}(\pi)>0]> (1-\alpha)$ and $\mathrm{Pr}[\mathrm{F}(\pi)\leq 0]<\alpha$. 
In consequence, $\mathbf{VaR}_{\alpha}[\mathrm{F}(\pi)]>0$.
		
\item  $l/m<\alpha$.  We create dummy scenario $\xi'$ under which the due date of each job equals $2n$. The probability of the dummy scenario is $\frac{m\alpha-l}{m-l}$. The probability of each of the remaining scenarios is equal to $\frac{1}{m}(1-\frac{m\alpha-l}{m-l})=\frac{1-\alpha}{m-l}$.
 Assume that the answer to \textsc{Min 3-Sat} is yes. 
 So, there is an assignment to the variables which satisfies at most $m-l$ clauses. By the construction, there is a schedule $\pi$ whose cost is positive under at most $m-l$ scenarios. Hence 
 $$\mathrm{Pr}[\mathrm{F}(\pi)\leq 0]=1-\mathrm{Pr}[\mathrm{F}(\pi)>0]\geq 1-(m-l)\frac{1-\alpha}{m-l}=\alpha$$
  and $\mathbf{VaR}_{\alpha}[\mathrm{F}(\pi)]\leq 0$. Assume that the answer to \textsc{Min 3-Sat} is no. Then for each assignment more than $m-l$ clauses are satisfied. By the construction, for every schedule $\pi$ there are more than $m-l$ scenarios under which the cost $\pi$ is positive. Therefore
  $\mathrm{Pr}[\mathrm{F}(\pi)>0]>(m-l)\frac{1-\alpha}{m-l}=(1-\alpha)$
  and $\mathrm{Pr}[\mathrm{F}(\pi)\leq 0]<\alpha$, so
   $\mathbf{VaR}_{\alpha}[\mathrm{F}(\pi)]>0$.
  \end{enumerate}
\end{proof}

It follows from Theorem~\ref{thm1}  that the problem discussed in~\cite{NBN17} is strongly NP-hard and not at all approximable even in the very restrictive case in which all job processing times are equal to~1.
It was shown in~\cite{KZ16d} that $\textsc{Min-Exp}~1|p_j=1|T_{\max}$ is strongly NP-hard and hard to approximate within $7/6-\epsilon$ for any $\epsilon>0$. Hence, we immediately get from Theorem~\ref{thmexp} that for each constant $\alpha\in [0,1)$, $\textsc{Min-CVaR}_\alpha~1|p_j=1|T_{\max}$ is strongly NP-hard and hard to approximate within $7/6-\epsilon$ for any $\epsilon>0$. 

We consider now the problem with the total tardiness criterion. The deterministic $1||\sum T_j$ problem is known to be NP-hard~\cite{LA77}. However, $1|p_j=1|\sum T_j$ is polynomially solvable(see, e.g.,~\cite{B07}). The following result characterizes the complexity of the minmax version of this problem:
\begin{thm}
	\textsc{Min-Max}~$1|p_j=1|\sum T_j$ is strongly NP-hard and not approximable within $\frac{5}{4}-\epsilon$ for any $\epsilon>0$.
	\label{tmm}
\end{thm}	
\begin{proof}
	We will show a reduction from the strongly NP-complete \textsc{3-Sat} problem, in which	
	 we are given boolean variables $x_1,\dots, x_n$, a collection of clauses $\mathcal{C}_1,\dots \mathcal{C}_m$, where each clause is a disjunction of at most $3$ literals (variables or their negations) and  
	 we ask if there is an assignment to the variables which satisfies all clauses  (see, e.g.,~\cite{GJ79}).
Given an instance of \textsc{3-Sat},
we create two jobs $J_{x_j}$ and $J_{\overline{x}_j}$ for each variable $x_j$, 
	$j\in [n]$, $|J|=2n$.
	A  due date scenario~$\xi_i$ corresponding to clause $\mathcal{C}_i=(l_1 \vee l_2 \vee l_3)$ is created in the same way as in the proof of Theorem~\ref{thm1}.
	Additionally, for each variable $x_j$ we create scenario $\xi_j'$ under which the due dates of $J_{x_j}$ and $J_{\overline{x}_j}$ are $2(j-1)+\frac{1}{2}$ and the due dates of the remaining jobs are set to $2n$ (see Table~\ref{tab1a}).
	We first show that the answer to \textsc{3-Sat} is yes if and only if there is a schedule $\pi$ such that $\max_{\xi\in \mathcal{U}} \sum_{j\in J} T_j(\pi, \xi)\leq 2$.	
	\begin{table}[ht]
	  \centering
	  \caption{The set of jobs and the due date scenarios for the formula $(x_1\vee \overline{x}_2 \vee \overline{x}_3)\wedge (\overline{x}_2 \vee \overline{x}_3 \vee x_4) \wedge (\overline{x}_1 \vee  x_2 \vee \overline{x}_4) \wedge (x_1 \vee x_2 \vee x_3) \wedge (x_1 \vee x_3 \vee \overline{x}_4)$. Schedule $\pi=(J_{x_1},J_{\overline{x}_1}, J_{\overline{x}_2}, J_{x_2}, J_{x_3}, J_{\overline{x}_3}, J_{\overline{x}_4}, J_{x_4}))$ corresponds to a truth assignment.} \label{tab1a}
	  \small
\begin{tabular}{l|lllll|lllll}
						  & $\xi_1$ & $\xi_2$ & $\xi_3$ & $\xi_4$ & $\xi_5$ & $\xi'_1$  & $\xi'_2$ & $\xi'_3$ & $\xi'_4$ \\ \hline
					$J_{x_1}$ & 1  & 2 & 2 & 1 & 1 & $\frac{1}{2}$ & 8 & 8 & 8\\
					$J_{\overline{x}_1}$  & 2 & 2 & 1 & 2 & 2 & $\frac{1}{2}$ & 8 & 8 & 8 \\  \hline
					$J_{x_2}$ &  4 & 4 & 3 & 3 & 4 & 8 & $2+\frac{1}{2}$ & 8 & 8 \\
					$J_{\overline{x}_2}$ & 3 & 3 & 4 & 4 & 4 & 8 & $2+\frac{1}{2}$ & 8 & 8\\  \hline
					$J_{x_3}$ & 6 & 6 & 6 & 5 & 5 & 8 &  8 & $4+\frac{1}{2}$ & 8 \\
					$J_{\overline{x}_3}$ & 5 & 5 & 6 & 6 & 6 & 8 & 8 & $4+\frac{1}{2}$ & 8 \\  \hline
					$J_{x_4}$ &  8 & 7 & 8 & 8 & 8 & 8 & 8 & 8 & $6+\frac{1}{2}$\\
					$J_{\overline{x}_4}$ & 8 & 8 & 7 & 8 & 7& 8 & 8 & 8 & $6+\frac{1}{2}$ \\  \hline
			\end{tabular}
		\end{table}
	
Assume that the answer to \textsc{3-Sat} is yes. Consider schedule $\pi=(J_1,J'_1,J_2,J'_2,\dots,J_n, J'_n)$, where $J_j,J'_j\in\{J_{x_j}, J_{\overline{x}_j}\}$. Furthermore $J_{x_j}$ is processed before $J_{\overline{x}_j}$ if and only if $x_j=1$. Since in every clause at least one literal is true, at most two jobs in $\pi$ are late under each scenario $\xi_i\in \mathcal{U}$. The tardiness of each job in $\pi$ under any $\xi_i\in \mathcal{U}$ is at most~1. Furthermore, the total tardiness in $\pi$ under any $\xi'_j$ is exactly~2.
	 In consequence,  $\max_{\xi\in \mathcal{U}} \sum_{j\in J} T_j(\pi, \xi)\leq 2$.
	
	Assume that there is a schedule $\pi$, such that $\max_{\xi\in \mathcal{U}} \sum_{j\in J} T_j(\pi, \xi)\leq 2$. We claim that $\pi=(J_1,J'_1,J_2,J'_2,\dots,J_n, J'_n)$, where $J_j,J'_j\in\{J_{x_j}, J_{\overline{x}_j}\}$. Suppose that this is not the case, and let $J_k$ ($J'_k$) be the last  job in $\pi$ which is not placed properly. The completion time of $J_k$ ($J'_k$) is at least $2k+1$. So, its tardiness under $\xi'_k$ is at least $2k+1-(2k-2+\frac{1}{2})=2.5$. Let $x_j=1$ if and only if $J_{x_j}$ is processed before $J_{\overline{x}_j}$ in $\pi$. Since only two jobs can be late under any $\xi_i$, this assignment satisfies all clauses and the answer to \textsc{3-Sat} is yes.
	
	In order to prove the lower approximation bound, it is enough to observe that if the answer to \textsc{3-Sat} is no, then each schedule has the total tardiness~3 under some scenario $\xi_i$ or~2.5 under some scenario $\xi'_j$, which gives a gap at least $\frac{5}{4}$.
\end{proof}

From the fact that $1||\sum T_j$ is weakly NP-hard (see~\cite{DL90}), we get immediately that more general $\textsc{Min-Exp}~1||\sum T_j$ problem is weakly NP-hard as well. The next theorem strengthens this result.
\begin{thm}
\label{thmsT}
	$\textsc{Min-Exp}~1||\sum T_j$ is strongly NP-hard.
\end{thm}
\begin{proof}
We will show a polynomial time reduction from the deterministic $1||\sum w_j T_j$ problem, which is known to be strongly NP-hard~\cite{LA77}. Consider  an instance of $1||\sum w_j T_j$.  Let $W=\sum_{j\in J} w_j>0$ and $P=\sum_{j\in J} p_j$.
	We build  an instance of $\textsc{Min-Exp}~1||\sum T_j$ with the same set of jobs $J$ and job processing times $p_j$, $j\in J$. 	
	We create $K=|J|=n$ due date scenarios as follows. Under scenario $\xi_j$, $j\in [n]$,  job $j$ has due date equal to $d_j$ and all the remaining jobs have due dates equal to $P$. We also fix ${\Pr}[\xi_i]=w_i/W$, $i\in [n]$. For any schedule $\pi$, we get
	${\bf E}[{\rm F}(\pi)]=\sum_{i\in [K]} {\rm Pr}[\xi_i] \sum_{j\in J} T_j(\pi, \xi_i)=\frac{1}{W}\sum_{i\in [n]} w_i \sum_{j\in J} T_j(\pi, \xi_i)$.
By the construction, we get $\sum_{j\in J} T_j(\pi, \xi_i)=[C_i(\pi)-d_i]^+$, so
${\bf E}[{\rm F}(\pi)]=\frac{1}{W}\sum_{i\in [n]} w_i [C_i(\pi)-d_i]^+$.
In consequence $1||\sum w_j T_j$ and $\textsc{Min-Exp}~1||\sum T_j$ have the same optimal solutions and the theorem follows.
\end{proof}

It was shown in~\cite{AAK11} that \textsc{Min-Max}~$1|p_j=1|\sum U_j$ with uncertain due dates is strongly NP-hard. The following theorem strengthens this result:

\begin{thm}
	 \textsc{Min-Max}~$1|p_j=1|\sum U_j$ is not approximable within any constant factor unless P=NP. 	\label{tnapr}
\end{thm}
\begin{proof}
Consider the following \textsc{Min-Max 0-1 Selection} problem. 
We are given a set of items $E=\{e_1,e_2,\dots,e_n\}$ and an integer $q\in [n]$.
 For each item $e_j$, $j\in [n]$, there is a cost $c_j(\xi_i)\in\{0,1\}$ under scenario $\xi_i$, $i\in [K]$. We seek a selection $X\subseteq E$ of exactly~$q$ items, $|X|=q$, which minimizes the maximum cost over all scenarios, i.e. the value of $\max_{i\in [K]} \sum_{e_i\in X} c_j(\xi_i)$. This problem was discussed in~\cite{KKZ13}, where it was shown that it is not approximable within any constant factor $\gamma\geq 1$.  We will show that there is a cost preserving reduction from \textsc{Min-Max 0-1 Selection} to the considered scheduling problem, which will imply the stated result.
 
 Given an instance of \textsc{Min-Max 0-1 Selection}, we build the corresponding instance of  \textsc{Min-Max}~$1|p_j=1|\sum U_j$  as follows.  We create a set of jobs $J=E$, $|J|=n$,  with deterministic unit processing times. For each $i\in [K]$, if $c_j(\xi_i)=1$ then $d_j(\xi'_i)=n-q$, and if $c_j(\xi_i)=0$, then $d_j(\xi'_i)=n$.  So, we create $K$ due date scenarios that correspond to the cost scenarios of \textsc{Min-Max 0-1 Selection}. 
 
 Suppose that there is a solution $X$ to \textsc{Min-Max 0-1 Selection} such that $\sum_{e_i\in X} c_j(\xi_i)\leq C$ for each $i\in [K]$.
 Hence $X$ contains at most $C$ items, $C\leq q$, with the cost equal to~1 under each scenario.  
 In the corresponding schedule $\pi$,
 we first process $n-q$ jobs from $J\setminus X$ and then the jobs in $X$ in any order. It is  easily seen that there are at most $C$ late jobs in $\pi$ under each scenario $\xi'_i$, hence the maximum cost of schedule $\pi$ is at most $C$.  Conversely, let $\pi$ be a schedule in which there are at most $C$ late jobs under each scenario $\xi'_i$.  Clearly $C\leq q$ since the first $n-q$ jobs in $\pi$ must be on-time in all scenarios. Let us form solution $X$ by choosing the items corresponding to the last $q$ jobs in $\pi$. Among these jobs at most $C$ are late under each scenario, hence the cost of $X$ is at most $C$ under each scenario $\xi_i$.

\end{proof}

Thus, by Theorem~\ref{thm01},  $\textsc{Min-CVaR}_\alpha~1|p_j=1|\sum U_j$ is strongly NP-hard for any $\alpha\in (0,1)$ 
(notice that 
 $p_{\max}=1$ in the proof of Theorem~\ref{thm01} and in the new scenario set $\mathcal{U}'$ still only due dates are uncertain).
\begin{thm}
\label{thmsU}
	$\textsc{Min-Exp}~1||\sum U_j$ is NP-hard.
\end{thm}
\begin{proof}
	Choose the deterministic $1||\sum w_j U_j$ problem, which is known to be NP-hard~\cite{KR74}. The reduction from this problem to $\textsc{Min-Exp}~1||\sum U_j$ is the same as the one in the proof of Theorem~\ref{thmsT}.
\end{proof}
It is worth noting that in the proof of Theorem~\ref{thmsU} we require an arbitrary probability distribution in the scenario set and we have  shown that
 the problem is only weakly NP-hard.  Its  complexity  for uniform probability distribution is open.
 
\subsection{Uncertain processing times}

In this section we characterize the complexity of the  problems under consideration when only processing times
 are uncertain. 
It has been shown in~\cite{KY97} that $\textsc{Min-Max}~1||\sum C_j$ is strongly NP-hard. Furthermore, this problem is also hard to approximate within $\frac{6}{5}-\epsilon$ for any $\epsilon>0$~\cite{MNO13}.  Using Theorem~\ref{thm01}, we can immediately conclude that the same negative result holds for $\textsc{Min-VaR}_{\alpha}~1||\sum C_j$ for any $\alpha\in (0,1]$. Also, strong NP-hardness of the min-max problem implies that $\textsc{Min-CVaR}_{\alpha}~1||\sum C_j$ is strongly NP-hard for each fixed $\alpha \in (0,1)$. Observe that the boundary case $\alpha=0$ (i.e. $\textsc{Min-Exp}~1||\sum C_j$) is polynomially solvable, as it easily reduces to the deterministic $1||\sum C_j$ problem. 
Since $1||\sum T_j$ is a special case of $1||\sum C_j$, with $d_j=0$ for each $j\in J$, the same negative results are true for the problem with the total tardiness criterion. Observe, however that $\textsc{Min-Exp}~1||\sum T_j$ is also NP-hard, since the deterministic $1||\sum T_j$ problem is known to be weakly NP-hard~\cite{DL90}.

It has been shown in~\cite{AC08} that \textsc{Min-Max}~$1||\sum U_j$ with uncertain processing times and deterministic due dates is NP-hard. The following theorem strengthens this result:

\begin{thm}
\label{thmsUp}
	\textsc{Min-Max}~$1||\sum U_j$ is strongly NP-hard. This assertion remains true even when all the jobs have  a common deterministic due date.
\end{thm}
\begin{proof}
We show a polynomial time reduction from the \textsc{3-Sat} problem (see the proof of Theorem~\ref{tmm}). Given an instance of \textsc{3-Sat}, we create an instance of \textsc{Min-Max}~$1||\sum U_j$ in the following way. For each variable $x_i$ we create two jobs $J_{x_i}$ and $J_{\overline{x}_i}$, so $J$ contains $2n$ jobs. The due dates of all these jobs are the same under each scenario and equal 2. For each clause $C_j=(l_1,l_2,l_3)$ we construct processing time scenario $\xi_i$, under which the jobs $J_{\overline{l}_1}, J_{\overline{l}_2}, J_{\overline{l}_3}$  have processing time equal to 1 and all the remaining jobs have processing times equal to 0. Then, for each pair of jobs $J_{x_i}, J_{\overline{x}_i}$ we construct scenario $\xi'_i$ under which the processing times of $J_{x_i}, J_{\overline{x}_i}$ are 2 and all the remaining jobs have processing times equal to 0.  A sample reduction is shown in Table~\ref{tab1b}.
We will show that the answer to \textsc{3-Sat} is yes if and only if there is a schedule $\pi$ such that $\max_{\xi \in \mathcal{U}} \sum_{j\in J} U(\pi,\xi)\leq n$.
\begin{table}[ht]
	  \centering
	  \caption{Processing time scenarios for the formula $(x_1\vee \overline{x}_2 \vee \overline{x}_3)\wedge (\overline{x}_2 \vee \overline{x}_3 \vee x_4) \wedge (\overline{x}_1 \vee  x_2 \vee \overline{x}_4) \wedge (x_1 \vee x_2 \vee x_3) \wedge (x_1 \vee x_3 \vee \overline{x}_4)$. Schedule $\pi=(J_{x_1},J_{\overline{x}_2},J_{x_3},J_{\overline{x}_4} | J_{\overline{x}_1}, J_{x_2}, J_{\overline{x}_3}, J_{x_4})$ corresponds to a satisfying truth assignment.} \label{tab1b}
			\begin{tabular}{l|lllll|llll|l}
										 & $\xi_1$ & $\xi_2$ & $\xi_3$ & $\xi_4$ & $\xi_5$ & $\xi'_1$ & $\xi'_2$ & $\xi'_3$ & $\xi'_4$ & $d_i$ \\ \hline
$J_{x_1}$ 		& 0  & 0 &  1 & 0 & 0 &  2 & 0 & 0 & 0 &  2\\
$J_{\overline{x}_1}$  & 1 & 0 & 0 &  1 &  1&  2 & 0 & 0 & 0 &  2\\  \hline
$J_{x_2}$			 &   1 &  1 & 0 & 0 & 0 & 0 &  2 & 0 & 0 & 2\\
$J_{\overline{x}_2}$   & 0 & 0 &  1 & 1 & 0 & 0 &  2 & 0 & 0 &  2\\   \hline
$J_{x_3}$ 		&  1 &  1 & 0 & 0 & 0 & 0 & 0 & 2 & 0 &  2\\
$J_{\overline{x}_3}$ & 0 & 0 & 0 &  1 &  1 & 0 & 0 &  2 & 0 &  2\\  \hline
$J_{x_4}$ 		&  0 & 0 &  1 & 0 & 1 & 0 & 0 & 0 &  2 & 2\\
$J_{\overline{x}_4}$ & 0 & 1 & 0 & 0 & 0 & 0 & 0 & 0 &  2 &  2\\  \hline
			\end{tabular}
		\end{table}

Assume that the answer to \textsc{3-Sat} is yes. Then there exists a truth assignment to the variables which satisfies all the clauses. Let us form schedule $\pi$ by processing first the jobs corresponding to true literals in any order and processing then the remaining jobs in any order. From the construction of the scenario set it follows that the completion time of the $n$th job in $\pi$ under each scenario is not greater than 2. In consequence, at most $n$ jobs in $\pi$ are late under each scenario and $\max_{\xi \in \mathcal{U}} \sum_{j\in J} U(\pi,\xi)\leq n$. 

Assume that there is a schedule $\pi$ such that $\sum_{j\in J} U(\pi,\xi)\leq n$ for each $\xi \in \mathcal{U}$, which means that at most $n$ jobs in $\pi$ are late under each scenario.  Observe first that $J_{x_i}$ and $J_{\overline{x}_i}$ cannot appear among the first $n$ jobs in $\pi$ for any $i\in [n]$; otherwise more than $n$ jobs would be late in $\pi$ under $\xi'_i$. Hence the first $n$ jobs in $\pi$ correspond to a truth assignment to the variables $x_1,\dots,x_n$, i.e. when $J_l$ is among the first $n$ jobs, then the literal $l$ is true. Since $f(\pi,\xi_i)\leq n$, the completion time of the $n$-th job in $\pi$ under $\xi_i$ is not greater than 2. We conclude that at most two jobs among the first $n$ job have processing time equal to 1 under $\xi_i$, so there are at most two false literals for each clause and the answer to \textsc{3-Sat} is yes.
\end{proof}

 We thus get from Theorem~\ref{thmsUp} that
	 $\textsc{Min-VaR}_{\alpha}~1||\sum U_j$ is strongly NP-hard for any $\alpha\in (0,1)$ and  $\textsc{Min-CVaR}_{\alpha}~1||\sum U_j$ is strongly NP-hard for any $\alpha\in (0,1]$. The boundary case with $\alpha=0$ (i.e. \textsc{Min-Exp}~$1||\sum U_j$ with uncertain processing times) is an interesting open problem.

\section{Positive complexity results}
\label{sec5}

In this section we establish some positive complexity results. Namely, we provide several polynomial and approximation algorithms for particular problems. A summary of the results can be found in Table~\ref{tabs3}.

\subsection{Problems with uncertain due dates}

Consider the $\textsc{Min-Exp}~1|p_j=1|\sum w_jU_j$ problem with uncertain due dates. We  introduce variables $x_{ij}\in \{0,1\}$, $i\in [n]$, $j\in [n]$, where $x_{ij}=1$ if $j\in [n]$ is the $i$th job in the schedule constructed. The variables satisfy the assignment constraints, i.e. $\sum_{i\in [n]} x_{ij}=1$ for each $j\in [n]$ and $\sum_{j\in [n]} x_{ij}=1$ for each $i\in [n]$. If $x_{ij}=1$, then the completion time of  job $j$ equals $i$. Define $c_{ijk}=w_j$ if $i>d_j(\xi_k)$ and $c_{ijk}=0$ otherwise, for each $i,j\in [n]$ and $k\in [K]$.
If the variables $x_{ij}$ describe $\pi$, then
$${\bf E}[{\rm F}(\pi)]=\sum_{k\in [K]} \sum_{i\in [n]}\sum_{j\in [n]}{\rm Pr}[\xi_k] c_{ijk}x_{ij}=\sum_{i\in [n]}\sum_{j\in [n]}c^*_{ij} x_{ij},$$
where 
$c^*_{ij}=\sum_{k\in [K]} {\rm Pr}[\xi_k] c_{ijk}$.
Hence the problem is equivalent to the \textsc{Minimum Assignment} with the cost matrix $c^*_{ij}$. The same result holds for $\textsc{Min-Exp}~1|p_j=1|\sum w_j T_j$. It is enough to define $c_{ijk}=w_j[i-d_j(\xi_k)]^+$ for $i,j\in [n]$, $k\in [K]$. We thus get the following  results:
\begin{thm}
	$\textsc{Min-Exp}~\mathcal{P}$
	is polynomially solvable, when
	 $\mathcal{P}\in\{1|p_j=1|\sum w_j U_j,\, 1|p_j=1|\sum w_j T_j\}$.
	\label{texppol}
\end{thm}
From Theorems~\ref{texppol} and~\ref{thmexp}, we  immediately get the following approximation result:
\begin{thm}
	$\textsc{Min-CVaR}_{\alpha}~\mathcal{P}$ is
	 approximable within $\rho=\min\{\frac{1}{{\rm Pr}_{\min}},\frac{1}{1-\alpha}\}$, when
	$\mathcal{P}\in\{1|p_j=1|\sum w_j U_j, \,1|p_j=1|\sum w_j T_j\}$.
\end{thm}
Since $\textsc{Min-Max}~1|p_j=1|\sum w_jU_j$ and $\textsc{Min-Max}~1|p_j=1|\sum w_jT_j$ are special cases of 
 the min-max version of  \textsc{Minimum Assignment}, which is approximable within~$K$~(see, e.g.,~\cite{ABV09}), 
 both problems are approximable within~$K$ as well.

We now study the $\textsc{Min-Exp}~1|| \sum w_j T_j$ problem with uncertain due dates and deterministic processing times. This problem is strongly NP-hard since $1||\sum w_j T_j$ is strongly NP-hard. The expected cost of $\pi$ can be rewritten as  ${\bf E}[{\rm F}(\pi)]=\sum_{j\in J} \sum_{i\in [K]} {\rm Pr}[\xi_i] [C_j(\pi)-d_j(\xi_i)]^+$.
 We thus get a single machine scheduling problem $1||\sum f_j$ with job-dependent cost
 functions of form $f_j(C_j(\pi))= \sum_{i\in [K]} {\rm Pr}[\xi_i] [C_j(\pi)-d_j(\xi_i)]^+$, $j\in J$.
 Note also that these functions are nonnegative and nondecreasing with respect to~$C_j(\pi)$. 
 The same analysis can be done for  the $\textsc{Min-Exp}~1|| \sum w_j U_j$ problem with uncertain due dates and deterministic processing times.
 Hence and from~\cite{CMSV17}, where a $(4+\epsilon)$-approximation algorithm, for any $\epsilon>0$, for this class of problems was provided,
 we  get the following result (see also Theorem~\ref{thmexp}):
  \begin{thm}
\label{tpdexp}
 If $\mathcal{P}\in \{1||\sum w_j U_j,\, 1||\sum w_j T_j\}$, then 
 $\textsc{Min-Exp}$~$\mathcal{P}$ is approximable within  $4+\epsilon$ and  $\textsc{Min-CVaR}_\alpha~\mathcal{P}$
 is approximable within  $\min\{\frac{4+\epsilon}{{\rm Pr}_{\min}},\frac{4+\epsilon}{1-\alpha}\}$, $\epsilon>0$,
 for any $\epsilon>0$ and each constant $\alpha\in [0,1)$
\end{thm}

When the probability distribution in $\mathcal{U}$ is uniform, then the approximation ratio in Theorem~\ref{thmexp} can be improved to $\min\{(4+\epsilon)K, \frac{4+\epsilon}{1-\alpha}\}$.
Since \textsc{Min-Max}~$\mathcal{P}$ is a special case of $\textsc{Min-CVar}_{\alpha}~\mathcal{P}$ with uniform probability distribution and $\alpha$ sufficiently large, we get that
$\textsc{Min-Max}$~$\mathcal{P}$, $\mathcal{P}\in \{1||\sum w_j U_j,\, 1||\sum w_j T_j\}$,
 is approximable within
$(4+\epsilon)K$ for any $\epsilon>0$.

\subsection{The total weighted flow time criterion}
 
In this section we focus on 
 the problems with the total weighted flow time criterion.
 We start by recalling  a well known property
 (see, e.g.,~\cite{MNO13}),
  which states that every such a problem with 
  uncertain processing times and deterministic weights can be transformed into an equivalent problem with uncertain weights and deterministic processing times. This transformation goes as follows. For each processing time scenario $\xi_i$, $i\in [K]$, we invert the role of processing times and weights obtaining the weight scenario $\xi'_i$. Formally, $p_j=w_j$ and $w_j(\xi'_i)=p_j(\xi_i)$ for each $i\in [K]$.  The new scenario set $\mathcal{U}'$ contains scenario $\xi'_i$  with ${\rm Pr}[\xi_i']={\rm Pr}[\xi_i]$ for each $i\in [K]$. We also invert the precedence constraints, i.e. if $i\rightarrow j$ in the original problem, then $j\rightarrow i$ in the new one. Given a feasible schedule $\pi=(\pi(1),\dots,\pi(n))$, let $\pi'=(\pi(n),\dots,\pi(1))$ be the corresponding inverted schedule. Of course, schedule $\pi'$ is feasible for the inverted precedence constraints. It is easy to verify that $f(\pi, \xi_i)=f(\pi',\xi'_i)$ for each $i\in [K]$. In consequence ${\bf CVaR}_\alpha[{\rm F}(\pi)]={\bf CVar}_\alpha[{\rm F'}(\pi')]$ and ${\bf VaR}_\alpha[{\rm F}(\pi)]={\bf VaR}_\alpha[{\rm F'}(\pi')]$, where ${\rm F'}(\pi')$ is the random cost of $\pi'$ for scenario set $\mathcal{U}'$. Hence, the original problem with uncertain processing times and the new one with uncertain weights have the optimal solutions with the same performance measure.

 From now on we make the assumption that 
  the jobs have deterministic processing times $p_j$, $j\in J$ and $w_j(\xi_i)$ is the weight of job $j$ under scenario $\xi_i$, $i\in [K]$. The value of ${\bf CVaR}_\alpha[{\rm F}(\pi)]$, for a fixed schedule $\pi$, can be computed by solving the following optimization problem (see the formulation (\ref{exVC1})b):
 \begin{equation}
\label{exWC1}
	\begin{array}{llll}
		 \min & \displaystyle \gamma + \frac{1}{1-\alpha}\sum_{i\in [K]} {\rm Pr}[\xi_k] u_k\\
		 \text{s.t.}& \gamma + u_k \geq  \displaystyle \sum_{j\in J}w_j(\xi_k) C_j(\pi) & k\in [K]\\
		 & u_k\geq 0 & k\in [K]
	\end{array}
\end{equation}

 Let $\delta_{ij}\in \{0,1\}$, $i,j\in [n]$, be binary variables such that $\delta_{ij}=1$ if job $i$ is processed before job $j$ in
 a schedule constructed. The vectors of all feasible job completion times $(C_1,\dots, C_n)$ can be described by the following system of constraints~\cite{PO80}:
	\begin{equation}
	\label{cCC2}
		 \begin {array}{llll}
				 VC: & C_j=p_j+\sum_{i\in J\setminus\{j\}} \delta_{ij} p_i & j\in J\\
				 &\delta_{ij}+\delta_{ji}=1 & i,j\in J, i\neq j \\
				 &\delta_{ij}+\delta_{jk}+\delta_{ki} \geq 1 & i,j,k \in J\\
				 &\delta_{ij}=1 & i\rightarrow j\\
				 &\delta_{ij}\in \{0,1\}& i,j \in J
		\end{array}
\end{equation}
Let us denote by $VC'$ the relaxation of $VC$, in which the constraints $\delta_{ij}\in \{0,1\}$ are relaxed with $0\leq \delta_{ij}\leq 1$.
It has been proved in~\cite{SH96b,HA97}  that each vector $(C_1,\dots, C_n)$ that satisfies $VC'$ also satisfies the following inequalities:
	\begin{equation}
		\label{Schin}
			\sum_{j\in I} p_jC_j\geq \frac{1}{2}\left((\sum_{j\in I} p_j)^2+\sum_{j\in I} p_j^2\right) \text{ for all } 
			I \subseteq J.
	\end{equation}
The formulations~(\ref{cCC2}) and (\ref{exWC1}) lead to the following mixed integer programming model for 
$\textsc{Min-CVaR}_{\alpha}~1|prec|\sum w_j C_j$ with uncertain weights:
\begin{equation}
\label{exWC3}
	\begin{array}{llll}
		 \min & \displaystyle \gamma + \frac{1}{1-\alpha}\sum_{i\in [K]} {\rm Pr}[\xi_k] u_k\\
		 \text{s.t.}& \gamma + u_k \geq \sum_{j\in J}w_j(\xi_k) C_j & k\in [K]\\
		 & \text{Constraints VC} \\
		 & u_k\geq 0 & k\in [K]
	\end{array}
\end{equation}

We now solve the relaxation of~(\ref{exWC3}), in which $VC$ is replaced with $VC'$.
  Let $(C_1^*, \dots, C_n^*)$ be the relaxed optimal job completion times and $z^*$ be the optimal value of the relaxation. 
Consider discrete random variable $\mathrm{Y}$, which takes the value $\sum_{j\in J} w_j(\xi_i)C^*_j$ with probability ${\rm Pr}[\xi_i]$, $i\in [K]$.
The equality $z^*={\bf CVaR}_\alpha[\mathrm{Y}]$ holds.
We  relabel the jobs so that $C^*_1\leq C^*_2\leq \cdots\leq\ C_n^*$ and form schedule $\pi=(1,2,\dots,n)$
 in nondecreasing order of~$C^*_j$.
Since the vector~$(C_j^*)$ satisfies~$VC'$ it must also satisfy~(\ref{Schin}).  Hence,  setting $I=\{1,\dots,j\}$, we get
	$$\sum_{i=1}^j p_iC^*_i\geq \frac{1}{2}\left((\sum_{i=1}^j p_i)^2+\sum_{i=1}^j p_i^2\right)\geq \frac{1}{2}\left((\sum_{i=1}^j p_i)^2\right).
	$$
	Since $C^*_j\geq C_i^*$ for each $i\in \{1\dots j\}$, we get
	$C^*_j\sum_{i=1}^j p_i\geq \sum_{i=1}^j p_iC^*_i \geq \frac{1}{2}(\sum_{i=1}^j p_i)^2$
	and, finally $C_j=\sum_{i=1}^j p_j \leq 2 C^*_j$ for each $j\in J$ -- this reasoning is 
	the same as in~\cite{SH96b}.

For each scenario $\xi_i\in \mathcal{U}$, the inequality
		$f(\pi,\xi_i)=\sum_{j\in J} w_j(\xi_i)C_j \leq 2 \sum_{j\in J} w_j(\xi_i)C^*_j $
		holds.
		 By Lemma~\ref{lemowa1}, we have
		${\bf CVaR}_\alpha[{\rm F}(\pi)]\leq 2\cdot {\bf CVaR}_\alpha[Y]=2z^*$.
		Since $z^*$ is a lower bound on the value of an optimal solution, $\pi$ is a 2-approximate schedule. Let us summarize the obtained result.
\begin{thm}
\label{thmcappr1}
	 $\textsc{Min-CVaR}_{\alpha}~1|prec|\sum w_j C_j$ is approximable within~2 for each $\alpha\in [0,1)$.
\end{thm}
This result can be refined when the deterministic $1|prec|\sum w_j C_j$ problem is polynomially solvable (for example, when the precedence constraints form an sp-graph, see, e.g.,~\cite{B07}). In this case $\textsc{Min-Exp}~1|prec|\sum w_j C_j$ is polynomially solvable, and we can also apply Theorem~\ref{thmexp}, which leads to the following result:
\begin{thm}
\label{thmcappr2}
	 If $1|prec|\sum w_j C_j$ is polynomially solvable, then $\textsc{Min-CVaR}_{\alpha}~1|prec|\sum w_j C_j$ is approximable within~$\min\{\frac{1}{1-\alpha}, 2\}$ for each $\alpha\in [0,1)$.
\end{thm}
Observe that $\frac{1}{1-\alpha}<2$ for each $\alpha<0.5$. Let us consider  $\textsc{Min-VaR}_{\alpha}~1|prec|\sum w_j C_j$ problem. The value of ${\bf VaR}_\alpha[{\rm F}(\pi)]$, for a fixed schedule $\pi$, can be computed by solving the following MIP problem (see~(\ref{exVC1})a):
\begin{equation}
\label{exVC123}
	\begin{array}{llll}
		 \min & \theta\\
		\text{s.t.} & \displaystyle \sum_{j\in J} w_j(\xi_k)C_j(\pi)-\theta\leq M_k \beta_k & k\in [K]\\
		 &\displaystyle  \sum_{k\in [K]} {\rm Pr}[\xi_i]  \beta_k \leq 1-\alpha\\
		 & \beta_k\in \{0,1\} & k\in [K]
	\end{array}
\end{equation}
where $M_k$ is an upper bound on the schedule cost under scenario $\xi_k$, $k\in [K]$. 
Using the formulation~(\ref{cCC2}) together with~(\ref{exVC1}), we can get a mixed integer programming formulation for $\textsc{Min-VaR}_{\alpha}~1|prec|\sum w_j C_j$. By replacing the constraints $VC$ with relaxed $VC'$ in the constructed model, we get a mixed integer problem with $K$ binary variables. This problem can be solved in polynomial time when $K$ 
is a constant. The same analysis as for  $\textsc{Min-CVaR}_{\alpha}~1|prec|\sum w_j C_j$ (we also use Lemma~\ref{lemowa1}) leads to the following result:
\begin{thm}
\label{thmcappr3}
	 If the number of scenarios is constant, then $\textsc{Min-VaR}_{\alpha}~1|prec|\sum w_j C_j$ is approximable within~2 for each $\alpha\in (0,1]$.
\end{thm}



\subsection{The bottleneck objective}

In this section we address a class of single machine  scheduling problems with a bottleneck objective, i.e. in which $f(\pi)=\max_{j\in J} f_j(C_j(\pi))$, where $f_j(t)$ is the cost of completing job $j$ at time $t$. An important and well known example is $1|prec| \max w_j T_j$, in which the maximum weighted tardiness is minimized. This problem can be solved in $O(n^2)$ time by Lawler's algorithm~\cite{LA73}. We will use the fact that the minmax versions of the bottleneck problems are polynomially solvable for a wide class of cost functions~\cite{BFSS16, KZ16d}. In particular, the minmax version of $1|prec| \max w_j T_j$ with uncertain processing times and uncertain due dates can be solved in $O(Kn^2)$ time by using the algorithm constructed in~\cite{KZ16d}. In the following, we will assume that $f(\pi, \xi)=\max_{j\in J} w_j T_j(\pi, \xi)$ for a given scenario $\xi\in \mathcal{U}$. We also assume that job processing times and due dates are nonnegative integers under all scenarios and job weights are positive integers. In consequence, the value of $f(\pi, \xi)$ is a nonnegative integer for each $\xi$.

 Let $f_{\max}$ be an upper bound on the schedule cost over all scenarios.
Let $h:\Qset_{+}^K\rightarrow \Qset_{+}$ be 
a nondecreasing function with respect to~$\Qset_{+}^K$. Suppose that $h$ can be evaluated in $g(K)$ time for a given vector $\pmb{t}=(t_1,\dots,t_K)\in \Zset_{+}^K$. 
Consider the corresponding scheduling problem $\mathcal{PS}$, in which we seek a feasible schedule $\pi\in \Pi$ minimizing $H(\pi)=h(f(\pi,\xi_1),\dots, f(\pi, \xi_K))$.  We can find such a schedule by solving a number of the following  \emph{auxiliary problems}:
given a vector $\pmb{t}\in\Zset_{+}^K$, 
check if $\Pi(\pmb{t})=\{\pi \in \Pi\,:\, 
f(\pi,\xi_i)\leq t_i, i\in [K]\}$ is nonempty,
and if so, return any schedule $\pi_{\pmb{t}}\in \Pi(\pmb{t})$. From the monotonicity of the function~$h$, it follows that for each 
$\pi \in \Pi(\pmb{t})$ the inequality $h(f(\pi,\xi_1),\ldots, f(\pi, \xi_K))\leq 
h(\pmb{t})$ is true. Thus,
 in order to solve the problem $\mathcal{PS}$, it suffices to
  enumerate all possible vectors $\pmb{t}=(t_1,\dots,t_K)$, where $t_i\in \{0,1,\dots,f_{\max}\}$, $i\in [K]$,
 and compute $\pi_{\pmb{t}}\in \Pi(\pmb{t})$ if $\Pi(\pmb{t})$ is nonempty. A schedule~$\pi_{\pmb{t}}$ with the minimum value of $H(\pi_{\pmb{t}})$ is returned.

The crucial step in this method is solving the auxiliary problem. 
We now show that this can be done in polynomial time for the bottleneck problem with the maximum weighted tardiness criterion.
Given any
$\pmb{t}\in\Zset_{+}^K$,  we first form scenario set $\mathcal{U}'$ by specifying the following parameters  for each $\xi_i \in \mathcal{U}$ and $j\in J$:
\[
p_j(\xi_i')=p_j(\xi_i), \; w'_j=1,\; 
d_j(\xi_i')=\max\{C\geq 0\,:\,w_j(C-d_j(\xi_i))\leq t_i\}=t_i/w_j+d_j(\xi_i).
\]
The scenario set $\mathcal{U}'$ can be built in $O(Kn)$ time.
We  then solve the minmax problem with  scenario set $\mathcal{U}'$,
which can be done in $O(Kn^2)$ time by using the algorithm constructed in~\cite{KZ16d}.
If the maximum cost of the schedule $\pi$ returned is 0, then $\pi_{\pmb{t}}=\pi$; otherwise $\Pi(\pmb{t})$ is empty.  
 Since all the risk criteria considered in this paper are nondecreasing functions 
 with respect to  schedule costs over scenarios (see Lemma~\ref{lemowa1} for $\gamma=1$)
 and $g(K)$ is negligible in comparison with $Kn^2$, we get the following result:
 \begin{thm}
$ \textsc{Min-Exp}~\mathcal{P}$,
 $\textsc{Min-VaR}_{\alpha}~\mathcal{P}$ and
 $\textsc{Min-CVaR}_{\alpha}~\mathcal{P}$ are 
 solvable in $O(f_{\max}^K (Kn^2))$ time, when $\mathcal{P}$ is $1|prec|\max w_j T_j$.
 \end{thm}
 The above running time is pseudopolynomial if $K$ is constant.
  Notice that the special cases, when $\mathcal{P}$ is $1|prec, p_j=1| T_{\max}$ are solvable in $O(Kn^{K+2})$ time, which is polynomial if $K$ is constant (as we can fix  $f_{\max}=n$).

We now show that the problems admit an FPTAS
if $K$ is a constant and $h(\gamma \pmb{t})\leq \gamma h( \pmb{t})$, 
for any $\pmb{t}\in\Qset_{+}^K$, $\gamma\geq 0$.
First we partition the interval $[0,f_{\max}]$ into geometrically increasing subintervals:
$[0,1)\cup \bigcup_{\ell \in [\eta]}[(1+\epsilon)^{\ell-1}, (1+\epsilon)^{\ell})$,
where $\eta=\lceil \log_{1+\epsilon} f_{\max} \rceil$ and $\epsilon\in (0,1)$.
Then we 
enumerate all possible vectors $\pmb{t}=(t_1,\dots,t_K)$, where 
$t_i\in \{0,1\}\cup   \bigcup_{\ell \in [\eta]}\{(1+\epsilon)^{\ell}\}$, $i\in [K]$,
 and find $\pi_{\pmb{t}}\in \Pi(\pmb{t})$ if $\Pi(\pmb{t})\not=\emptyset$. 
 Finally,
 we output a schedule~$\pi_{\hat{\pmb{t}}}$ that minimizes  value of $H(\pi_{\pmb{t}})$ 
 over the nonempty subsets of schedules.
Obviously,  the running time is
$O(( \log_{1+\epsilon} f_{\max})^K (Kn^2+g(K)))=
O((\epsilon^{-1} \log f_{\max})^K (Kn^2+g(K)))$.
Let $\pi^*$ be an optimal schedule.
Fix $\ell_i\in \{0,1,\ldots,\eta\}$  for each $i\in [K]$,
 such that $(1+\epsilon)^{\ell_i-1}\leq f(\pi^*,\xi_i)<(1+\epsilon)^{\ell_i}$,
 where
we assume that  $(1+\epsilon)^{\ell_i-1}=0$ for $\ell_i=0$.
This clearly forces
$\Pi((1+\epsilon)^{\ell_1},\dots,(1+\epsilon)^{\ell_K})\not=\emptyset$.
Moreover, $(1+\epsilon)^{\ell_i}\leq (1+\epsilon) f(\pi^*,\xi_i)$ for $\ell_i$, $i\in [K]$.
By the definition of~$\pi_{\hat{\pmb{t}}}$, we get
$H(\pi_{\hat{\pmb{t}}})\leq h((1+\epsilon)^{\ell_1},\dots,(1+\epsilon)^{\ell_K})$.
Since $h$ is a nondecreasing function and $h(\gamma \pmb{t})\leq \gamma h( \pmb{t})$,
$h((1+\epsilon)^{\ell_1},\dots,(1+\epsilon)^{\ell_K})\leq
(1+\epsilon)h(f(\pi^*,\xi_1),\dots,f(\pi^*,\xi_K))$.
Hence, $H(\pi_{\hat{\pmb{t}}})\leq (1+\epsilon)H(\pi^*)$.
By Lemma~\ref{lemowa1}, the risk criteria satisfy the additional assumption on the function $h(\pmb{t})$.  This leads to the following theorem:
\begin{thm}
$ \textsc{Min-Exp}~\mathcal{P}$,
 $\textsc{Min-VaR}_{\alpha}~\mathcal{P}$ and
 $\textsc{Min-CVaR}_{\alpha}~\mathcal{P}$ admit an FPTAS, when $\mathcal{P}$ is $1|prec|\max w_j T_j$ and the number of scenarios is constant.
 \end{thm}
 
\section{Conclusions and open problems}

In this paper we have discussed a wide class of single machine scheduling problems with uncertain job processing times and due dates. This uncertainty is modeled by a discrete scenario set with  a known probability distribution. In order to compute a solution we have applied the risk criteria, namely, the value at risk and conditional value at risk. The expectation and the maximum criteria are special cases of the risk measures. We have provided a number of negative and positive complexity results for problems with basic cost functions. 
Moreover, we have sharpened  some negative ones obtained in~\cite{AAK11,AC08}.
The picture of the complexity is presented in Tables~\ref{tabs1}-\ref{tabs3}.
Obviously,
 the negative results obtained
  remain true for more general cases, for instance, for the problems with more than one machine.

There is still a number of interesting open problems on the models discussed. The negative results for uncertain due dates assume that the number of due dates scenarios is a part of input. The complexity status of the problems when the number of due date scenarios is fixed (in particular, equals~2) is open. For uncertain processing times, an interesting open problem is $\textsc{Min-Exp}~1||\sum U_j$ (see Table~\ref{tabs2}). There is still a gap between the positive and negative results, in particular, we conjecture that the negative results for $\textsc{Min-Var}~\mathcal{P}$ for uncertain processing times (see Table~\ref{tabs2}) can be strengthen. Now they are just the same as for the $\textsc{Min-Max}~\mathcal{P}$.

\subsubsection*{Acknowledgements}
This work was  supported by
 the National Center for Science (Narodowe Centrum Nauki), grant  2017/25/B/ST6/00486.



\begin{thebibliography}{10}

\bibitem{AAK11}
H.~Aissi, M.~A. Aloulou, and M.~Y. Kovalyov.
\newblock Minimizing the number of late jobs on a single machine under due date
  uncertainty.
\newblock {\em Journal of Scheduling}, 14:351--360, 2011.

\bibitem{ABV09}
H.~Aissi, C.~Bazgan, and D.~Vanderpooten.
\newblock Min-max and min-max regret versions of combinatorial optimization
  problems: a survey.
\newblock {\em European Journal of Operational Research}, 197:427--438, 2009.

\bibitem{AC08}
M.~A. Aloulou and F.~D. Croce.
\newblock Complexity of single machine scheduling problems under scenario-based
  uncertainty.
\newblock {\em Operations Research Letters}, 36:338--342, 2008.

\bibitem{NBN17}
S.~Atakan, K.~Bulbul, and N.~Noyan.
\newblock Minimizng value-at-risk in single machine scheduling.
\newblock {\em Annals of Operations Research}, 248:25--73, 2017.

\bibitem{AV00}
I.~Averbakh.
\newblock Minmax regret solutions for minimax optimization problems with
  uncertainty.
\newblock {\em Operations Research Letters}, 27:57--65, 2000.

\bibitem{AZ02}
A.~Avidor and U.~Zwick.
\newblock Approximating {MIN $k$-SAT}.
\newblock {\em Lecture Notes in Computer Science}, 2518:465--475, 2002.

\bibitem{BFSS16}
N.~Brauner, F.~Gerd, S.~Yakov, and S.~Dzmitry.
\newblock Lawler's minmax cost algorithm: optimality conditions and
  uncertainty.
\newblock {\em Journal of Scheduling}, 19:401--408, 2016.

\bibitem{B07}
P.~Brucker.
\newblock {\em Scheduling Algorithms}.
\newblock Springer Verlag, Heidelberg, 5th edition, 2007.

\bibitem{ZSZD17}
Z.~Chang, S.~Song, Y.~Zhang, J.-Y. Ding, and R.~Chiong.
\newblock Distributionally robust single machine scheduling with risk aversion.
\newblock {\em European Journal of Operational Research}, 256:261--274, 2017.

\bibitem{CMSV17}
M.~Cheung, J.~Mestre, D.~B. Shmoys, and J.~Verschae.
\newblock A {P}rimal-{D}ual {A}pproximation {A}lgorithm for {M}in-{S}um
  {S}ingle-{M}achine {S}cheduling {P}roblems.
\newblock {\em SIAM Journal on Discrete Mathematics}, 31:825--838, 2017.

\bibitem{DK95}
R.~L. Daniels and P.~Kouvelis.
\newblock Robust scheduling to hedge against processing time uncertainty in
  single-stage production.
\newblock {\em Management Science}, 41:363--376, 1995.

\bibitem{DR16}
M.~Drwal and R.~Rischke.
\newblock Complexity of interval minmax regret scheduling on parallel identical
  machines with total completion time criterion.
\newblock {\em Operations Research Letters}, 44:354--358, 2016.

\bibitem{DL90}
J.~Du and J.~Y.-T. Leung.
\newblock Minimizing {Total} {Tardiness} on {One} {Machine} is {NP}-hard.
\newblock {\em Mathematics of Operations Research}, 15:483--495, 1990.

\bibitem{GJ79}
M.~R. Garey and D.~S. Johnson.
\newblock {\em Computers and {I}ntractability. A {G}uide to the {T}heory of
  {NP}-{C}ompleteness}.
\newblock W. H. Freeman and Company, 1979.

\bibitem{HA97}
L.~A. Hall, A.~S. Schulz, D.~B. Shmoys, and J.~Wein.
\newblock Scheduling to minimize average completion time: off-line and on-line
  approximation problems.
\newblock {\em Mathematics of Operations Research}, 22:513--544, 1997.

\bibitem{KR74}
R.~M. Karp.
\newblock Reducibility {A}mong {C}ombinatorial {P}roblems.
\newblock In {\em Complexity of Computer Computations}, pages 85--103, 1972.

\bibitem{K05}
A.~Kasperski.
\newblock Minimizing maximal regret in the single machine sequencing problem
  with maximum lateness criterion.
\newblock {\em Operations Research Letters}, 33(4):431--436, 2005.

\bibitem{KKZ13}
A.~Kasperski, A.~Kurpisz, and P.~Zieli{\'n}ski.
\newblock Approximating the min-max (regret) selecting items problem.
\newblock {\em Information Processing Letters}, 113:23--29, 2013.

\bibitem{KZ14s}
A.~Kasperski and P.~Zieli{\'n}ski.
\newblock Minmax (regret) sequencing problems.
\newblock In F.~Werner and Y.~Sotskov, editors, {\em Sequencing and scheduling
  with inaccurate data}, chapter~8, pages 159--210. Nova Science Publishers,
  2014.

\bibitem{KZ16d}
A.~Kasperski and P.~Zieli{\'n}ski.
\newblock Single machine scheduling problems with uncertain parameters and the
  {OWA} criterion.
\newblock {\em Journal of Scheduling}, 19:177--190, 2016.

\bibitem{KM94}
R.~Kohli, R.~Krishnamurti, and P.~Mirchandani.
\newblock The minimum satisfiability problem.
\newblock {\em SIAM Journal on Discrete Mathematics}, 7:275--283, 1994.

\bibitem{KY97}
P.~Kouvelis and G.~Yu.
\newblock {\em Robust Discrete Optimization and its Applications}.
\newblock Kluwer Academic Publishers, 1997.

\bibitem{KPU02}
P.~Krokhmal, J.~Palmquist, and S.~P. Uryasev.
\newblock Portfolio optimization with conditional value-at-risk objective and
  constraints.
\newblock {\em Journal of Risk}, 4:43--68, 2002.

\bibitem{LA73}
E.~L. Lawler.
\newblock Optimal sequencing of a single machine subject to precedence
  constraints.
\newblock {\em Management Science}, 19:544--546, 1973.

\bibitem{LA77}
E.~L. Lawler.
\newblock A pseudopolynomial algorithm for sequencing jobs to minimize total
  tardiness.
\newblock {\em Annals of Discrete Mathematics}, 1:331--342, 1977.

\bibitem{AL06}
V.~Lebedev and I.~Averbakh.
\newblock Complexity of minimizing the total flow time with interval data and
  minmax regret criterion.
\newblock {\em Discrete Applied Mathematics}, 154:2167--2177, 2006.

\bibitem{MNO13}
M.~Mastrolilli, N.~Mutsanas, and O.~Svensson.
\newblock Single machine scheduling with scenarios.
\newblock {\em Theoretical Computer Science}, 477:57--66, 2013.

\bibitem{MSU99}
R.~H. M{\"{o}}hring, A.~S. Schulz, and M.~Uetz.
\newblock Approximation in stochastic scheduling: the power of {LP}-based
  priority policies.
\newblock {\em Journal of the ACM}, 46:924--942, 1999.

\bibitem{NST15}
K.~Natarajan, D.~Shi, and K.-C. Toh.
\newblock A probabilistic model for minmax regret in combinatorial
  optimization.
\newblock {\em Operations Research}, 62:160--181, 2014.

\bibitem{NIK10}
E.~Nikolova.
\newblock Approximation algorithms for offline risk-averse combinatorial
  optimization.
\newblock In {\em Proceedings of APPROX'10}, 2010.

\bibitem{O12}
W.~Ogryczak.
\newblock {R}obust {D}ecisions under {R}isk for {I}mprecise {P}robabilities.
\newblock In Y.~Ermoliev, M.~Makowski, and K.~Marti, editors, {\em {M}anaging
  {S}afety of {H}eterogeneous {S}ystems}, pages 51--66. Springer-Verlag, 2012.

\bibitem{P00}
G.~C. Pflug.
\newblock Some remarks on the {V}alue-at-{R}isk and the {C}onditional
  {V}alue-at-{R}isk.
\newblock In S.~P. Uryasev, editor, {\em Probabilistic {C}onstrained
  {O}ptimization: {M}ethodology and {A}pplications}, pages 272--281. Kluwer
  Academic Publishers, 2000.

\bibitem{PI08}
M.~Pinedo.
\newblock {\em Scheduling. Theory, Algorithms, and Systems}.
\newblock Springer, 2008.

\bibitem{PO80}
C.~N. Potts.
\newblock An algorithm for the single machine sequencing problem with
  precedence constraints.
\newblock {\em Mathematical Programming Study}, 13:78--87, 1980.

\bibitem{RU00}
R.~T. Rockafellar and S.~P. Uryasev.
\newblock Optimization of conditional value-at-risk.
\newblock {\em The Journal of Risk}, 2:21--41, 2000.

\bibitem{SSL14}
S.~Sarin, H.~Sherali, and L.~Liao.
\newblock Minimizing conditional-value-at-risk for stochastic scheduling
  problems.
\newblock {\em Journal of Scheduling}, 17:5--15, 2014.

\bibitem{SH96b}
A.~S. Schulz.
\newblock Scheduling to minimize total weighted completion time: Performance
  guarantees of {LP-Based} heuristics and lower bounds.
\newblock In {\em IPCO}, pages 301--315, 1996.

\bibitem{SSU16}
M.~Skutella, M.~Sviridenko, and M.~Uetz.
\newblock Unrelated {M}achine {S}cheduling with {S}tochastic {P}rocessing
  {T}imes.
\newblock {\em Mathematics of Operations Research}, 41:851--864, 2016.

\bibitem{SU05}
M.~Skutella and M.~Uetz.
\newblock Stochastic {M}achine {S}cheduling with {P}recedence {C}onstraints.
\newblock {\em SIAM Journal on Computing}, 34:788--802, 2005.

\bibitem{YA88}
R.~R. Yager.
\newblock On ordered weighted averaging aggregation operators in multi-criteria
  decision making.
\newblock {\em IEEE Transactions on Systems, Man and Cybernetics}, 18:183--190,
  1988.

\bibitem{YK93}
G.~Yu and P.~Kouvelis.
\newblock Complexity results for a class of min-max problems with robust
  optimization applications.
\newblock In P.~M. Pardalos, editor, {\em Complexity in Numerical
  Optimization}. World Scientyfic, 1993.

\end{thebibliography}


\end{document}